\documentclass[proceedings,submission]{dmtcs} 
\usepackage[latin1]{inputenc}
\usepackage{subfigure}
\usepackage[round]{natbib}
\usepackage{amsmath,amssymb}
\usepackage{xcolor} 
\usepackage{footnote}

\usepackage[math]{cellspace}
\usepackage{hyperref}
\definecolor{darkgreen}{rgb}{0,0.4,0}
\definecolor{BrickRed}{rgb}{0.65,0.08,0}
\definecolor{darkblue}{rgb}{0.0, 0.0, 0.55}
\hypersetup{colorlinks=true,linkcolor=blue,citecolor=darkblue,filecolor=BrickRed,urlcolor=darkgreen}

\graphicspath{{./}{pics/}}

\newcommand{\walksym}{\omega}
\newcommand{\walk}[1]{\walksym_{#1}}
\newcommand{\PNgeq}{P_0^{\geq}}
\newcommand{\dNgeq}{\delta_0^{\geq}}
\newcommand{\rNgeq}{\rho_0^{\geq}}
\newcommand{\pzeroi}{p_{0,i}}
\newcommand{\stepset}{\mathcal{S}}

\newcommand{\LandauO}{\mathcal{O}}
\newcommand{\Landauo}{o}
\newcommand{\PR}{\mathbb{P}} 
\newcommand{\E}{\mathbb{E}} 
 
\newcommand{\Ac}{\mathcal{A}}
\newcommand{\Ec}{\mathcal{E}}

\newcommand{\N}{\mathbb{N}}
\newcommand{\Q}{\mathbb{Q}}
\newcommand{\Z}{{\mathbb Z}}

\newcommand{\diamondend}{ \hfill ${\lozenge} $}
\newtheorem{theo}{Theorem}[section]
\newtheorem{lemma}[theo]{Lemma}
\newtheorem{prop}[theo]{Proposition}
\newenvironment{definition}[1][]{\refstepcounter{theo} \medskip \noindent \textbf{\textit{Definition \thetheo #1:}} }{ \diamondend }

\newenvironment{remarkFormulaEnd}[1][]{\refstepcounter{theo}\medskip \noindent\textbf{\textit{Remark \thetheo #1:}} }{ }

\author[Cyril Banderier \& Michael Wallner]{Cyril Banderier\addressmark{1} \and Michael Wallner\addressmark{2} }

\title[Some reflections on directed lattice paths.]{Some reflections on directed lattice paths.}

\address{\addressmark{1}Laboratoire d'Informatique de Paris Nord, UMR CNRS 7030, Universit\'e Paris Nord, 93430 Villetaneuse, France\\
\addressmark{2}Institute of Discrete Mathematics and Geometry, TU Wien, Wiedner Hauptstr. 8-10/104, A-1040 Wien, Austria}

\begin{document}
\maketitle
\begin{center}
\it{ This article corresponds, up to minor typo corrections and a correction of Table~4 (half-normal instead of Rayleigh), to the extended abstract which appeared in the Proceedings of the AofA'14 Paris Conference.}
\end{center}

\begin{abstract}
\paragraph{Abstract.}
This article analyzes \emph{directed lattice paths}, when a boundary reflecting or absorbing condition is added to the classical models. 
The lattice paths are characterized by two time-independent sets of rules (also called steps) 
which have a privileged direction of increase and are therefore essentially one-dimensional objects. 
Depending on the spatial coordinate, one of the two sets of rules applies, 
namely one for altitude $0$ and one for altitude bigger than $0$. 
The abscissa $y=0$ thus acts as a border which either absorbs or reflects steps. 
The absorption model corresponds to the model analyzed by Banderier and Flajolet (``Analytic combinatorics of directed lattice paths''),
while the reflecting model leads to a more complicated situation.
We show how the generating functions are then modified:
the kernel method strikes again but here it unfortunately does not give a nice product formula. 
This makes the analysis more challenging, and, in the case of {\L}ukasiewicz walks,
we give the asymptotics for the number of excursions, arches and meanders.
Limit laws for the number of returns to 0 of excursions are given.
We also compute the limit laws of the final altitude of meanders.
The full analytic situation is more complicated than the Banderier--Flajolet model 
(partly because new ``critical compositions'' appear, forcing us to introduce new key quantities, like the drift at 0),
and we quantify to what extent the global drift, and the drift at 0 play a role in the ``universal'' behavior of such walks. 
\end{abstract}
\vspace{-1ex}
\section{Introduction}
\label{sec:intro}
In Brownian motion theory, many possible boundary conditions for a Brownian-like process have been considered (e.g.~absorption, killed Brownian motion, reflected Brownian motion... see~\cite{Feller54}).
Solving a stochastic differential equation with a reflecting boundary condition is known as the 
Skorokhod problem (see~\cite{Skorokhod62}). Such models appear e.g.~in queueing theory (see~\cite{Kingman62}).
In this article, we want to investigate properties of a discrete equivalent 
of such models, namely directed lattice paths in $\Z^2$, having a reflecting boundary at $y=0$.

If one considers lattice paths which are ``killed'' or ``absorbed'' at $y=0$,
then this is equivalent to the model analyzed in~\cite{bafl02}.
In what follows, we want to compare the basic properties (exact enumeration, asymptotics, limit laws) 
of these two discrete models (absorption versus reflection).
In particular, we will consider the {\L}ukasiewicz paths (defined hereafter), 
which are present in numerous fields like analysis of algorithms, combinatorics, language theory, probability theory and 
biology. 
This broad applicability is due to a bijection with simple families of trees, see e.g.~\cite{memo78}. 
The enumerative and analytic properties of such lattice paths were considered
in \cite{bagi06} where limit laws for the area beneath {\L}ukasiewicz paths are derived, 
and also in \cite{brip11} where they are used to model polymers in chemistry,
or e.g.~in \cite{bani10}, which tackles the problem of enumeration and asymptotics of such walks of bounded height. 

Our key tools will be the kernel method and analytic combinatorics (see~\cite{flaj09}).
However, as we will see, the situation is more complicated in the case of a reflecting boundary:
first, bad luck, one does not have a nice product formula for the generating function anymore (unlike the absorption model),
second, the drift still plays a key role, but also does a ``second'' drift at 0,
and last but not least, several simultaneous singular behaviors can happen.
We first begin with a few definitions:

 \begin{definition} \label{def:LP}
 A {\it step set} $\stepset \subset \integers^2$, is a finite set of vectors $\{ (a_1,b_1), \ldots, (a_m,b_m)\}$. 
An $n$-step \emph{lattice path} or \emph{walk} is a sequence of vectors $v = (v_1,\ldots,v_n)$, such that $v_j$ is in $\stepset$. 
Geometrically, it is a set of points $\walksym =(\walk{0},\walk{1},\ldots,\walk{n})$ where $\walk{i} \in \integers^2, \walk{0} = (0,0)$ 
and $\walk{i}-\walk{i-1} = v_i$ for $i=1,\ldots,n$.
 The elements of $\stepset$ are called \emph{steps} or \emph{jumps}. 
The \emph{length} $|\walksym|$ of a lattice path is its number $n$ of jumps. 
 \end{definition}

 \begin{table}[htbp]
 \small
 \begin{center}\renewcommand{\tabcolsep}{3pt}
 \begin{tabular}{|c|c|c|}
 \hline
 & ending anywhere & ending at 0\\
 \hline
 \begin{tabular}{c} unconstrained \\ (on~$\Z$) \end{tabular}
 & \begin{tabular}{c} 

 {\includegraphics[width=3.5cm,height=12mm]{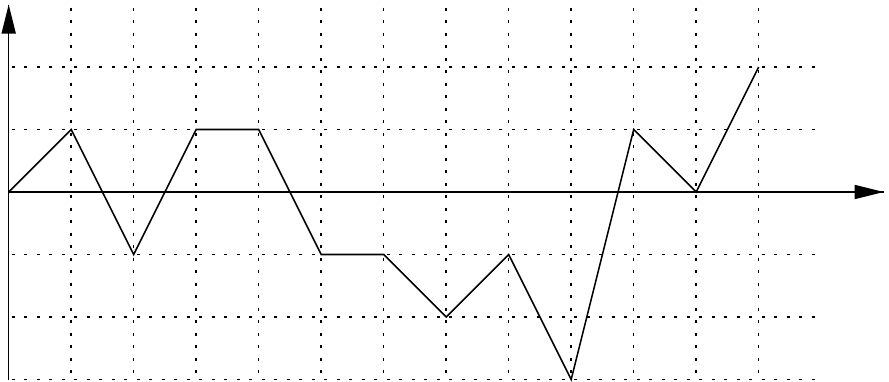}} \\ 
 walk/path ($\cal W$) 
 \end{tabular}
 & \begin{tabular}{c} 

 {\includegraphics[width=3.5cm,height=12mm]{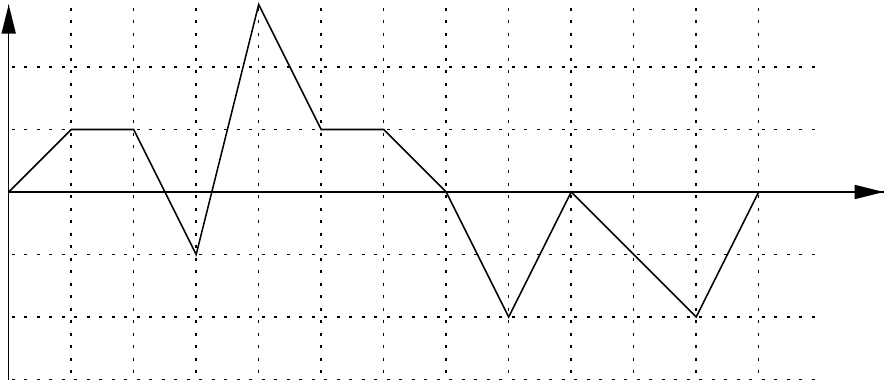}} \\
 bridge ($\cal B$)
 \end{tabular} \\
 \hline
 \begin{tabular}{c}constrained\\ (on $\N$) \end{tabular}
 & \begin{tabular}{c} 
 \includegraphics[width=3.5cm,height=12mm]{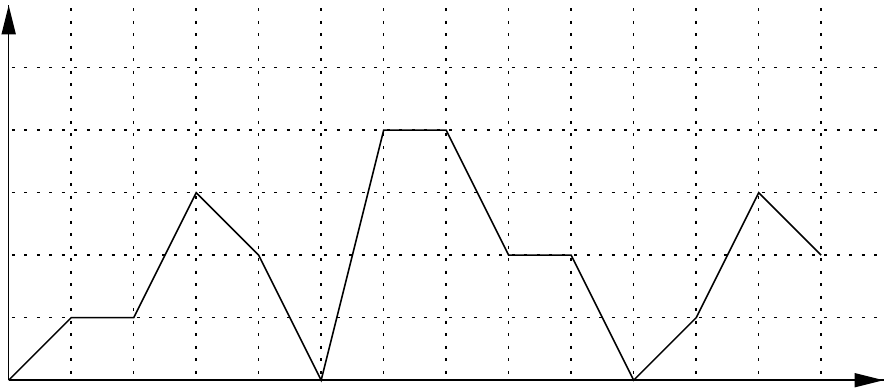} \\ 
 meander ($\cal M$)\\ 
 \end{tabular}
 & \begin{tabular}{c} 
 {
 \includegraphics[width=3.5cm,height=12mm]{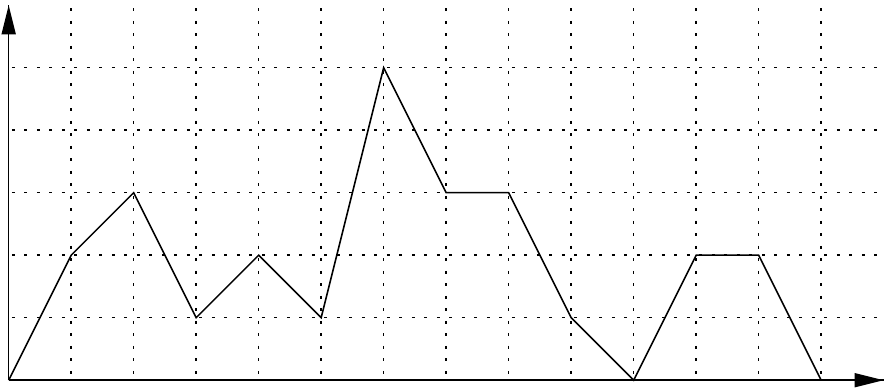}} \\ 
 excursion ($\cal E$)\\ 
 \end{tabular}\\
 \hline
 \end{tabular}
 \end{center}
 \caption{\label{fig-4types} 
 The four types of paths: walks, bridges, meanders and excursions.}
 \end{table}
We restrict our attention to \emph{directed paths} which are defined by the fact that for $(a,b) \in \stepset$ one must have $a > 0$. 
However, we will focus only on the subclass of \emph{simple paths}, where every element in the step set $\stepset$ is of the form $(1,b)$. 
In other words, these walks constantly move one step to the right, 
thus they are essentially unidimensional objects. We introduce the abbreviation $\stepset = \{ b_1, \ldots, b_n \}$ in this case. 
A \emph{{\L}ukasiewicz path} is a simple path where its associated step set $\stepset$ is a subset of $\{-1,0,1,\ldots\}$ and $-1 \in \stepset$.

 \begin{definition}
 For a given step set $\stepset = \{s_1,\ldots,s_m\}$, we define the respective {\it system of weights}\index{lattice path!weights} 
as $ \{w_1,\ldots,w_m\}$ with $w_j >0$ the associated weight to step $s_j$ for $j=1,\ldots,m$. 
The {\it weight of a path} is defined as the product of the weights of its individual steps. 
 \end{definition}

\smallskip
 This article mainly builds on the work done in \cite{bafl02}. Therein, the class of directed lattice paths 
in $\Z^2$ (under the absorption model) was investigated thoroughly by means of analytic combinatorics (see \cite{flaj09}). 
First, in Section \ref{sec:general}, the reflection-absorption model and the general framework are introduced. 
The needed bivariate generating function is defined and the governing functional equation is derived and solved:
here the ``kernel method'' plays the most significant role in order to obtain the generating function 
(as typical for many combinatorial objects which are recursively defined with a ``catalytic parameter'', see \cite{BMJ06}). 
In Section \ref{sec:Luk}, we turn our attention to {\L}ukasiewicz paths, 
and the asymptotic number of excursions is given. 
In Section \ref{sec:contacts}, the limit laws for the number of returns to zero of excursions are derived. 
In Section \ref{sec:meanders}, we establish the asymptotics of meanders. 
Section \ref{sec:finAlt} gives the asymptotics for the expected final altitude of meanders. 
\pagebreak
\section{Generating functions}
\label{sec:general}

Let us consider directed walks on $\N^2$, with a weighted step set $\stepset$, 
starting at the origin, confined to the upper half plane,
and which have another weighted step set $\stepset_0$ on the boundary $y=0$. 
All such walks are called \emph{meanders}, and the meanders ending on the abscissa are called \emph{excursions}. 

This walk model is thus encoded by two \emph{characteristic polynomials}: 
$P(u)$ and $P_0(u)$ are Laurent polynomials describing the allowed jumps when the walk is at altitude $k>0$ or $k=0$, respectively. 
We fix $c,d,c_0,d_0 \in \N$ and introduce the following notations:
\vspace{-1mm}
\begin{align*}
 P(u) &= \sum_{i=-c}^d p_i u^i, &
 P_0(u) &= \sum_{i=-c_0}^{d_0} \pzeroi u^i, & \PNgeq(u) &= \sum_{i=0}^{d_0} \pzeroi u^i.
\end{align*}

\vspace{-2mm}
In order to exclude trivial cases we require $p_c, p_d \neq 0$. These weights are probabilities, which means $p_i, \pzeroi \geq 0$ and $P(1)=P_0(1)=1$.
These step polynomials characterize the \emph{reflection-absorption model}:
depending of the chosen weights, the boundary behaves like a reflecting or an absorbing wall. 
We talk about a \emph{reflection model} if $\PNgeq(u)=P_0(u)$,
while we talk about an \emph{absorption model} if $\PNgeq(u)\neq P_0(u)$. 

\renewcommand{\arraystretch}{1.35}
\newcommand{\scalingDyckPics}{0.08}
\begin{savenotes}
\begin{table}[hbtp]
 \begin{center}
 \begin{tabular}{|>{\centering\arraybackslash} m{\scalingDyckPics\textwidth}|c|c|c|c|}
 \hline & bridges, & absolute value & excursions, & excursions, \\
 Dyck path & uniform model & of bridges
 & reflection model & absorption model \\
 \hline\hline {\includegraphics[width=\scalingDyckPics\textwidth]{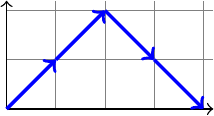}} & $\frac{1}{6}$ & \begin{tabular}{c}$\frac{1}{3} $\end{tabular} & $\frac{1}{3}$
\footnote{Note that the absolute value and the reflection model are in general not equivalent if the jumps (with their weights) are not symmetric: Let $P(u) = pu + q u^{-1}$ and $P_0(u) = p_0 u + q_0 u^{-1}$,
 then the probability of this first path (which is $1/3$ when $p=q=1/2$ and $p_0=1$) is $1/(1+q_0/q+p_0/p)$ in the absolute value model, while it is $p/(1+p)$ in the reflection model.}
& $\frac{1}{2}$ \\
 \hline {\includegraphics[width=\scalingDyckPics\textwidth]{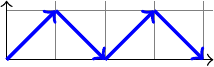}} & $\frac{1}{6}$ & $\frac{2}{3}$ & $\frac{2}{3}$ & $\frac{1}{2}$ \\
 \hline \includegraphics[width=\scalingDyckPics\textwidth]{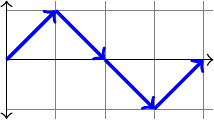} & $\frac{1}{6}$ & $0$ & $0$ & $0$ \\
 \hline \includegraphics[width=\scalingDyckPics\textwidth]{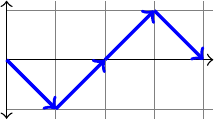} & $\frac{1}{6}$ & $0$ & $0$ & $0$ \\
 \hline \includegraphics[width=\scalingDyckPics\textwidth]{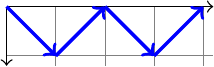} & $\frac{1}{6}$ & $0$ & $0$ & $0$ \\
 \hline \includegraphics[width=\scalingDyckPics\textwidth]{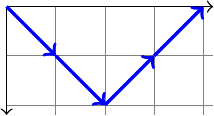} & $\frac{1}{6}$ & $0$ & $0$ & $0$ \\
 \hline
 \end{tabular}
 \end{center}
 \caption{Different constraints on the boundary $y=0$ lead to different probabilistic models. 
We give the probabilities of Dyck paths of length $4$ in the uniform, absolute value, reflection, and absorption model. 
From this table, one can already see one paradox associated
to the reflection model: one may think that the ``reflection'' 
will make the walk go far away. However, this is in part counterbalanced
by the fact that 0 has a ``heavier'' weight in this model (no loss of mass here,
contrary to the absorption model). Accordingly, 
there will be some interplay between the boundary,
the drift of the walk and the drift at 0.
We quantify this in our next sections.}
 \label{tab:lukDiffModels}
\end{table}
\renewcommand{\arraystretch}{1.0}
\end{savenotes}

\pagebreak
We define the generating function for our meanders to be 
\begin{align*}
 F(z,u) &:= \sum_{n,k \geq 0} F_{n,k} u^k z^n 
 = \sum_{n \geq 0} f_n(u) z^n = \sum_{k \geq 0} F_k(z) u^k,
\end{align*}
where the polynomials $f_n(u)$ describe the possible positions after $n$ steps,
and where $F_k(z)$ are the generating functions of walks starting at 0 
and ending at altitude $k$.

\begin{theo}[Generating function for meanders and excursions]
 \label{theo:lukBGF}
 The bivariate generating function of meanders (where $z$ marks size and $u$ marks final altitude) in the reflection-absorption model is algebraic:
 \begin{align}
 \label{eq:lukBGFalgebraic}
 F(z,u) &= \frac{1 - z \sum_{k=0}^{c-1} r_k(u) F_k(z)}{1-zP(u)},
 \end{align}
 where $r_k$ is a Laurent polynomial given by $r_k(u) = \sum_{j=-c}^{-k-1} p_j u^{j+k}$ for $k>0$ and $r_0(u)=P(u)-\PNgeq(u)$.
 Furthermore, the $F_k$'s are algebraic functions belonging to $\Q(u_1,\dots,u_c,p_{-c},\dots,$ $p_d,p_0^0,\dots,p_d^0,z)$, where the $u_i$'s are the roots of the equation $1-zP(u)=0$, such that $\lim_{z \to 0} u_i(z) = 0$. 
The $F_k$'s can be made explicit, e.g.~the generating function for excursions is
 \newcommand{\PROD}{\prod_i^c u_i}
 \begin{align}\label{F0}
 F_0(z) &= \frac{ \sum_{\ell = 1} ^{c} (-1)^{\ell+1} u_\ell^{c-1} V(\ell) } {\sum_{\ell = 1} ^{c} (-1)^{\ell+1} u_{\ell}^{c-1} \left( 1 - z\PNgeq(u_\ell) \right) V(\ell) },
 \end{align}
 where $V(\ell) = \prod_{\substack{ 1 \leq m < n \leq c \\ m \neq \ell,~n \neq \ell }}\left( u_m - u_n \right)$.
\end{theo}

\begin{proof}[(Sketch)]
 It is straightforward to derive a recurrence relation, by a step-by-step approach:
 \begin{align}
 f_0(u) &= 1, & 
 f_{n+1}(u) &= \{u^{\geq 0}\} \left[ P(u) \{u^{>0}\} f_n(u) + P_0(u) \{u^0\} f_n(u) \right], \label{eq:lukRecBGF}
 \end{align}
where $\{u^{> 0}\}$ extracts all the monomials of positive degree in $u$.
This recurrence leads to the following functional equation
 \begin{align}
 F(z,u) &= 1 + zP(u) F(z,u) - z\{u^{<0}\}P(u)F(z,u) - zF_0(z) \left(\{u^{\geq 0}\}P_0(u) - \{u^{\geq 0}\} P(u)\right), \label{eq:genFuncEq}
 \end{align} 
 \begin{align}
 F(z,u)(1-zP(u)) &= 1 - z \left(P(u) - \PNgeq(u)\right) F_0(z) - z \sum_{k=1}^{c-1} r_k(u) F_k(z). \label{eq:genFuncEq2}
 \end{align}

 The main tool for solving the functional equation is the \emph{kernel method}, which consists of binding $z$ and $u$ 
 in such a way that the left hand side vanishes. 
 From the theory of Newton--Puiseux expansions, we know that the kernel equation $1-zP(u)=0$ has $c+d$ distinct solutions, 
 with $c$ of them being called ``small branches'', as they map $0$ to $0$ and are in modulus smaller 
 than the other $d$ ``large branches'' which grow to infinity while approaching $0$. We call the small branches $u_1, \ldots, u_c$.

 Inserting the $c$ small branches into \eqref{eq:genFuncEq2},
 we get a linear system of $c$ equations in $c$ unknowns $F_0,\ldots,F_{c-1}$:
 $$\begin{cases}
 u_1^c - z \sum_{k=0}^{c-1} u_1^c r_k(u_1) F_k(z) &= 0, \notag \\
 \qquad \qquad \qquad \vdots & \label{eq:genLS} \\
 u_c^c - z \sum_{k=0}^{c-1} u_c^c r_k(u_c) F_k(z) &= 0. \notag
 \end{cases}$$

The system is non-singular, as can be proven via the local behavior of the $u_i$'s (and therefore of the $r_k(u_j)$'s) for $z\sim 0$. 
Using Cramer's formula on this matrix does not give (directly) a nice formula ; it is better to use first column subtraction for the $r_k$'s in the corresponding determinants,
and then expanding with respect to the $r_0(u_i)$ column leads to a nicer formula for each $F_i$ via Vandermonde-like formul\ae.
\end{proof}

Let us recall that the generating functions $\widetilde E(z)$ of the excursions in the absorption model of Banderier--Flajolet (i.e.,~when $P_0(u)=P(u)$)
is given by $\widetilde E(z)=(-1)^{c+1} (\prod u_i(z)) / (z p_{-c}) $, see \cite[Equation (20)]{bafl02}.
It is interesting to compare this simple formula with the more cumbersome formula that one gets for 
the generating functions $E(z)$ of the excursions in our reflection-absorption model (coming from a rewriting of \eqref{F0}): 
\begin{equation} E(z):=F_0(z) =\frac{\widetilde E(z)}{1-z \widetilde E(z) \sum_{i=1}^c r_0(u_i) u_i^{c-1}/V(i)}\,. \label{eq:Etild} \end{equation}
In one sense, this formula quantifies to what extent 
the border ``perturbs'' the former $\widetilde E(z)$ to lead to our new $E(z)$.
We now investigate the analytic counterparts of this perturbation.

\section{Asymptotics of excursions}
\label{sec:Luk}
From now on, we are going to work with \emph{aperiodic {\L}ukasiewicz paths}. By these, we understand paths with one jump of size $1$ down and finitely but arbitrarily many jumps up. Aperiodic means that there is no $p>1$ 
and there exists no polynomial $H(u)$ s.t.~$P(u) = u^{-1} H(u^p)$. Thus, the step polynomial is given as
\begin{align*}
 P(u) &= p_{-1} u^{-1} + p_0 + p_1 u + \ldots + p_d u^d,
\end{align*}
with $p_{-1} + \ldots + p_d = 1$ and $p_i \in [0,1]$. 
Since $c=1$, the linear system derived from \eqref{eq:genFuncEq2} transforms into
\begin{align*}
 1 + z \left(\PNgeq(u_1) - P(u_1) \right) F_0(z) = 0.
\end{align*}
We use the kernel equation $1-zP(u_1)=0$ to derive the generating function of excursions:
\begin{align}
 \label{eq:lukA0}
 E(z) := \sum e_n z^n &:= F_0(z) = \frac{1}{1 - z\PNgeq(u_1(z))}.
\end{align}

This nice formula has a natural combinatorial interpretation as $\operatorname{Seq}\left(z \PNgeq \left(\widetilde E(z) p_{-1} z\right)\right)$,
i.e.,~an excursion (in the reflection model) is a sequence of arches (i.e.,~an excursion touching 0 just at its two ends), and each arch begins with a positive jump $+k$, which has to be compensated
by $k$ excursions (well, shifted excursions: from altitude $j$ to altitude $j$, for $j$ from $1$ to $k$,
thus not touching 0, and thus in bijection with excursions, counted by $\widetilde E(z)$ and defined above formula \eqref{eq:Etild}) followed each by a $-1$ jump.

In \cite[Equation (42)]{bafl02}, it was shown that the principal branch $u_1(z)$ possesses the following asymptotic expansion for $z \to \rho^-$, where $\rho$ is the structural radius defined as $\rho=\tfrac{1}{P(\tau)}$ and $\tau > 0$ is the unique root of $P'(\tau)=0$ (note that $P$ is a convex function):
\begin{align}
 \label{eq:asyu1}
 u_1(z) &= \tau - \sqrt{2 \frac{P(\tau)}{P''(\tau)}} \sqrt{1 - z/\rho} + \LandauO(1-z/\rho), \qquad \text{ for } z \to \rho^-.
\end{align}

As this expansion will appear repeatedly in the sequel, we define $C:=\sqrt{2 \frac{P(\tau)}{P''(\tau)}}$.
The singularities of \eqref{eq:lukA0} depend on the roots of the denominator and on the singular behavior of $u_1(z)$, that is why we need the following lemma:
\begin{lemma}[Singularity of the denominator]
 \label{lem:lukA0Denom}
 Let $u_1(z)$ be the unique small branch of the kernel equation $1-zP(u)=0$. Then the equation $1 - z\PNgeq(u_1(z))=0$ has at most one solution in $z \in (0,\rho]$, which we denote by $\rho_1$.
\end{lemma}
\begin{proof}
 The functions $u_1(z)$ and $\PNgeq(u)$ are increasing on $[0,\rho)$ (see \cite{bafl02}). Figure~\ref{fig:lukSingularBehavior} shows the three possible configurations. The naming convention is adopted from its use in functional composition schemes in \cite[Chapter VI.9]{flaj09}.
\end{proof}

\begin{figure}[ht]
 \begin{center}
 \subfigure[supercritical case]{\includegraphics[width=0.32\textwidth]{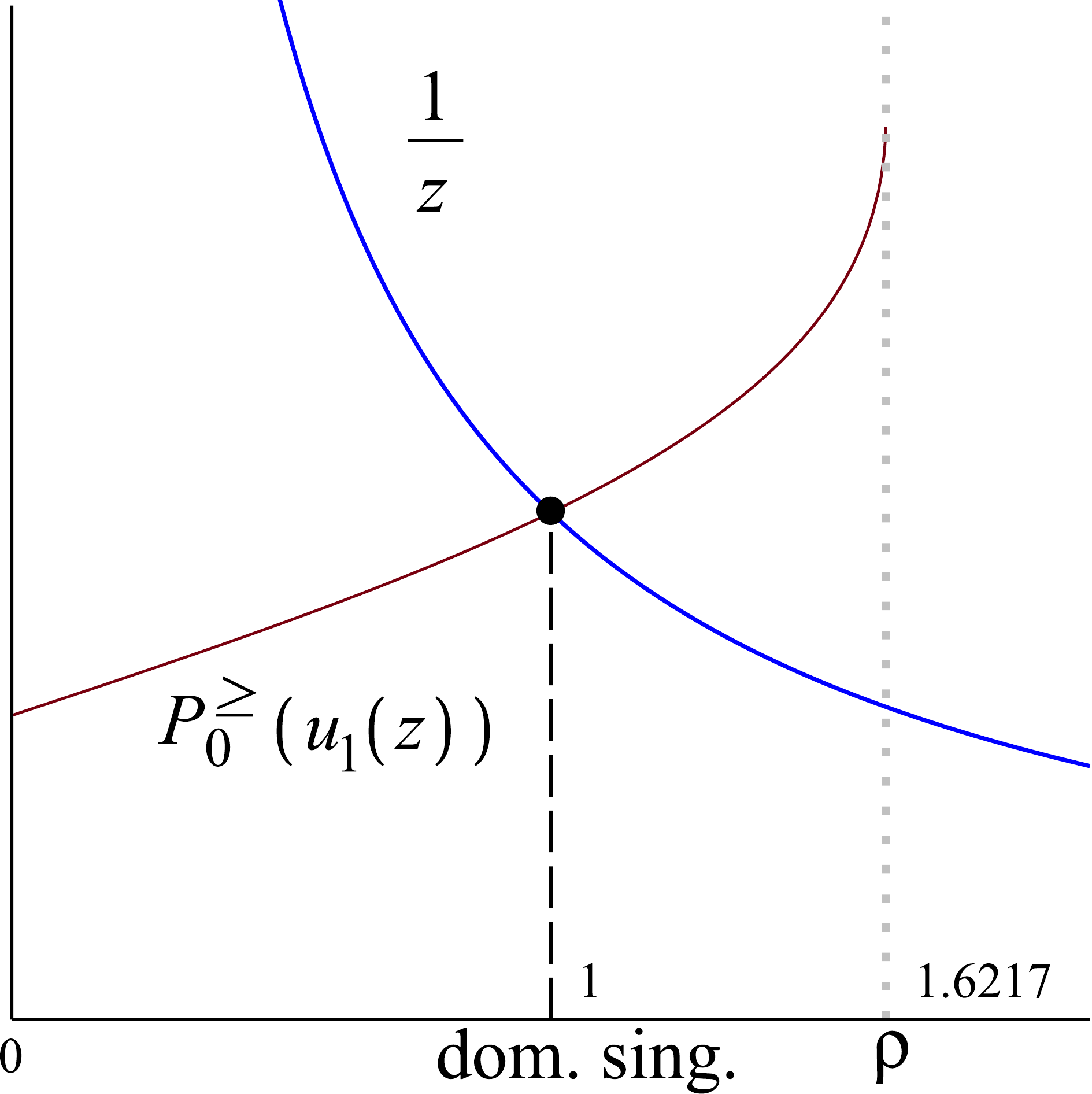}}
 \hfil
 \subfigure[critical case]{\includegraphics[width=0.32\textwidth]{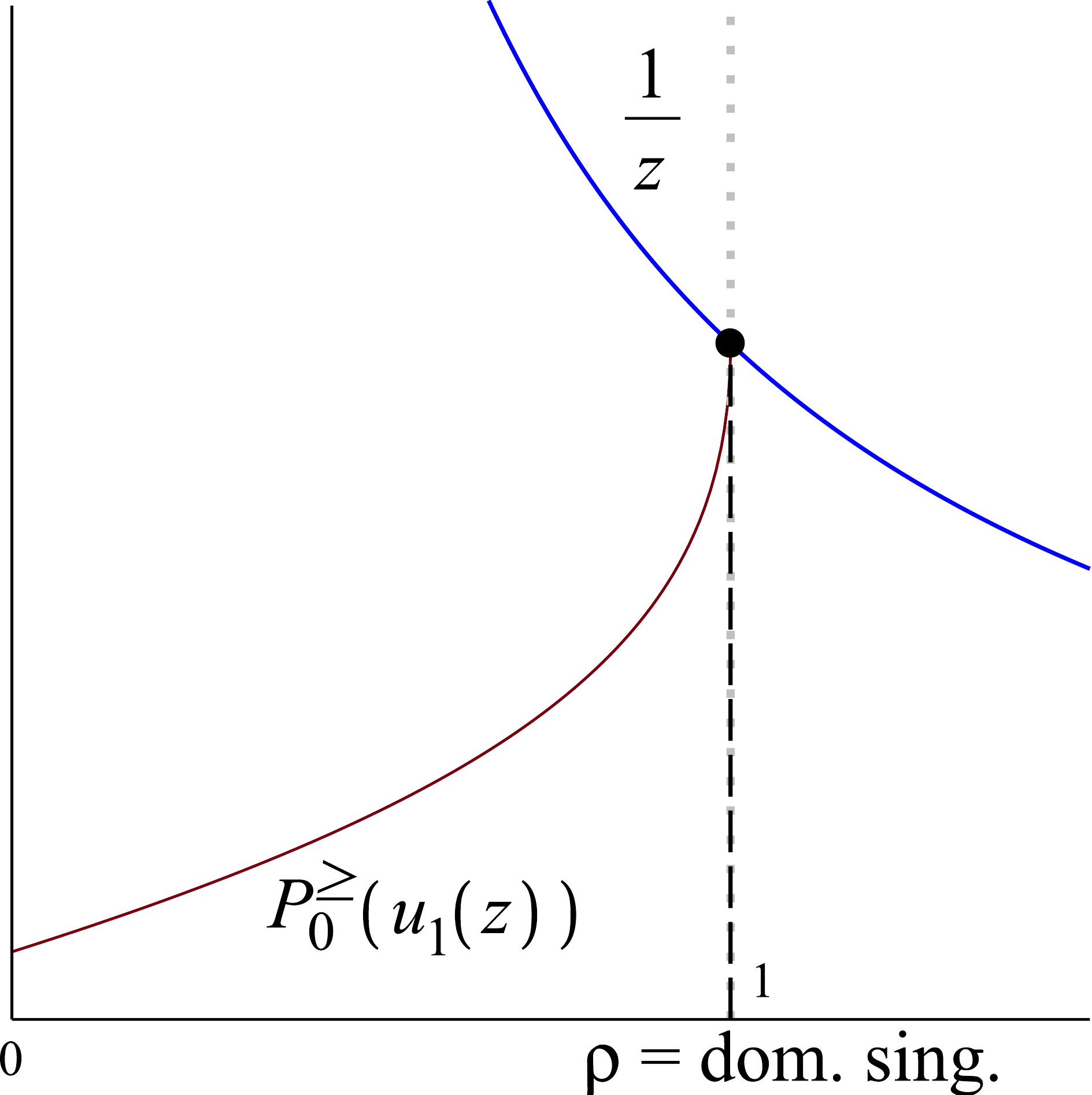}}
 \hfil
 \subfigure[subcritical case]{\includegraphics[width=0.32\textwidth]{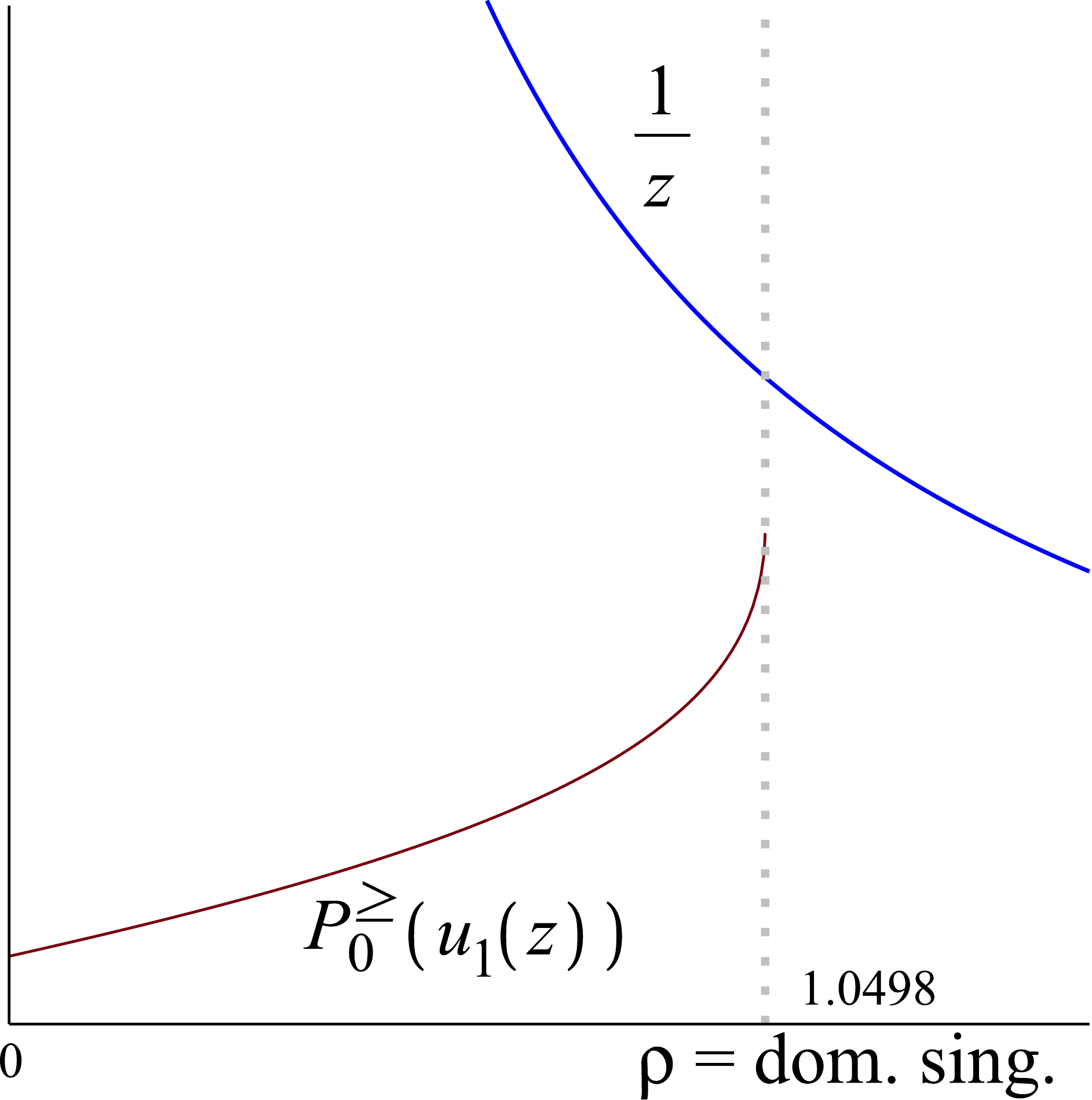}}
 \caption{Different singular behaviors of the generating function for the number of excursions. The increasing function represents $\PNgeq(u_1(z))$ where the decreasing function is $1/z$. The dotted line is at abscissa $\rho$ and the dashed line marks the dominant singularity. The latter is either located at the intersection or at $\rho$.}
 \label{fig:lukSingularBehavior}
 \end{center}
\end{figure}

\begin{theo}[Asymptotics of excursions]
 \label{theo:lukExcursions}
 Let $\tau$ be the structural constant determined by $P'(\tau)=0$, $\tau>0$, let $\rho = 1/P(\tau)$ be the structural radius and $\rho_1$ defined as in Lemma~\ref{lem:lukA0Denom}.
Define the constants $\alpha = \left. (\PNgeq(u_1(z)))' \right|_{z=\rho_1}$, $\gamma=\frac{1}{\alpha \rho_1^2+1}$, and $\kappa = C \rho (\PNgeq)'(\tau)$.
 The excursions in the reflection-absorption model possess the following asymptotic expansion:
 \begin{align}
 \label{eq:lukExcursionsFinal}
 E(z) &= 
 \begin{cases}
 \gamma (1-z/\rho_1)^{-1} + \LandauO(1), & \text{ supercritical case: } P(\tau)<\PNgeq(\tau),\\
 \frac{1}{\kappa} (1-z/\rho)^{-1/2} + \LandauO(1), & \text{ critical case: } P(\tau)=\PNgeq(\tau),\\
 E(\rho) - E(\rho)^2 \kappa (1-z/\rho)^{1/2} + \LandauO(1-z/\rho), & \text{ subcritical case: } P(\tau)>\PNgeq(\tau).
 \end{cases}
 \end{align}
\end{theo}

\begin{proof} [(Sketch)]
 We investigate $E(z)$ by means of singularity analysis. Therefore, three different cases are distinguished as the dominant singularity depends on the singular behavior (as illustrated by Figure~\ref{fig:lukSingularBehavior}). 
Then, the analysis of the corresponding Puiseux expansion yields the result.
\end{proof}

\section{Limit laws for the number of returns to zero}
\label{sec:contacts}

An \emph{arch} is defined as an excursion of size $>0$ whose only contact with the $x$-axis is at its end points. We denote the set of arches by $\Ac$. Every excursion (set $\Ec$) consists of a sequence of arches, i.e.,~$\Ec = \text{SEQ}(\Ac)$.

The symbolic method (see e.g.~\cite{flaj09}) directly provides the functional equation
\begin{align}
 \label{eq:arch2exc}
 E(z) = \frac{1}{1-A(z)},
\end{align}
which is easily solved to give the generating function for arches $A(z) = 1 - \tfrac{1}{E(z)}$.

\begin{prop}[Asymptotics of arches]
For a {\L}ukasiewicz walk, the number of arches satisfies
 \begin{align}
 [z^n] A(z) \underset{n \to \infty }{\sim}& \frac{\kappa}{2 \sqrt{\pi n^3}}\,,
\text{\qquad where $\kappa = C \rho (\PNgeq)'(\tau)$.}
 \end{align}
\end{prop}

\begin{proof}[(Sketch)]
 Define $\lambda := \tfrac{\PNgeq(\tau)}{P(\tau)} = \tfrac{\rho}{\rNgeq}$, then by \eqref{eq:lukA0} we get for $z \to \rho-$ that 
 $ A(z) = z \PNgeq(u_1(z)) 
 = \lambda - \kappa \sqrt{1-z/\rho} + \LandauO(1-z/\rho)$.
\end{proof}

A \emph{return to zero} is a vertex of a path of altitude $0$ whose abscissa is positive, i.e.,~the number of returns to zero is the number of times the abscissa is touched again after leaving the origin. 
In order to count the number of returns to zero of excursions of fixed size $n$, we can reverse the construction above for the generating function of arches. The generating function for excursions with exactly $k$ returns to zero is equal to $A(z)^{k}$. 
As stated in \cite{bafl02}, for any fixed $k$, this function also has a singularity of the square root type and is amenable to singularity analysis. Hence, we are able to derive the probability $\PR_{n,k}$ that a random excursion of size $n$ has exactly $k$ returns to zero for any fixed $k$:
\begin{align}
 \label{eq:probCon}
 \PR_{n,k} := \PR[\text{size}=n,~\text{\# returns to zero}=k] = \frac{[z^n]A(z)^{k}}{[z^n] E(z)}.
\end{align}
Let $X_n$ be the random variable for the number of arches among all excursions of size $n$. 
Note that $X_n$ also represents the returns to zero of a random excursion of size $n$. 

\begin{theo}[Limit laws for returns to zero]
 \label{theo:lawContacts}
 Additionally to the previously used constants $\alpha, \gamma$ and $\kappa$, we define $\alpha_2 = \left. (\PNgeq(u_1(z)))'' \right|_{z=\rho_1}$. 
 The number $X_n$ of returns to zero of a random excursion of size $n$ admits a limit distribution:
 \begin{enumerate}
  \item In the supercritical case, i.e.,~$P(\tau) < \PNgeq(\tau)$, 
 \begin{align*}
 \frac{X_n - \mu n}{\sigma \sqrt{n}}, \qquad \mu = \gamma, \qquad \sigma = \alpha_2 (\rho_1 \gamma)^3 - \gamma + \gamma^2(\rho_1+2)-2\gamma ^3,
 \end{align*}
 converges in law to a Gaussian variable $N(0,1)$.
  \item In the critical case, i.e.,~$P(\tau) = \PNgeq(\tau)$, the normalized random variable
 $\frac{\kappa}{\sqrt{2 \pi}} (X_n-1)$,
 converges in law to a Rayleigh distribution defined by the density $x e^{-x^2/2}$.
  \item In the subcritical case, i.e.,~$P(\tau) > \PNgeq(\tau)$, the limit distribution of $X_n-1$ is a discrete limit law,
namely the negative binomial distribution $\operatorname{NegBin}(2, \lambda)$, with $\lambda=\rho/\rho_0^{\geq}$:
 $${\mathbb P}(X_n-1=k) \sim (k+1) \lambda^k (1-\lambda)^2\,.$$
 \end{enumerate}
\end{theo}

\section{Asymptotics of meanders}
\label{sec:meanders}

A \emph{meander} is the natural generalization of an excursion, as it is defined as a directed walk confined to the upper half plane. 
We want to investigate the number of meanders or equivalently the ratio of meanders among all walks. This is a way to measure the effect of removing the constraint of ending on the $x$-axis.

\begin{theo}[Asymptotics of meanders]
 \label{theo:lukMeanders}
 Consider {\L}ukasiewicz walks in the absorption model. 

 The asymptotic behavior of the ratio of meanders of size $n$ is given in Table \ref{tab:lukMeanders}.
 \begin{table}[htbp]
 \begin{center}
 \begin{tabular}{|c||Sc|Sc|Sc|}
 \hline $[z^n]M(z) \sim$ & $\delta < 0$ & $\delta = 0$ & $\delta > 0$ \\ 
 \hline\hline Supercritical & \begin{tabular}{c}$\displaystyle \frac{\rho_1 \gamma}{E(1)(\rho_1-1)} \rho_1^{-n}$\end{tabular} 
 & --- 
 & \\
 \cline{1-3} Critical & $\displaystyle \frac{\rho}{E(1) \kappa(\rho-1)} \frac{\rho^{-n}}{\sqrt{\pi n}}$ & --- & $1-(1-\PNgeq(1) ) E(1)$\\
 \cline{1-3} Subcritical & $\displaystyle\frac{E(\rho)^2}{E(1)} \frac{\kappa \rho}{2(\rho-1)} \frac{\rho^{-n}}{\sqrt{\pi n^3}}$ & $\displaystyle \frac{E(1) \kappa}{\sqrt{\pi n}}$ & \\
 \hline
 \end{tabular}
 \end{center}
 \caption{Asymptotic ratio of meanders in the absorption model ($\PNgeq(1)<1$) with the structural constant $\tau >0$, $P'(\tau)=0$, the structural radius $\rho = 1/P(\tau)$ and the drift $\delta = P'(1)$. The constant $\rho_1$ is defined in Lemma \ref{lem:lukA0Denom}, whereas $\gamma$ and $\kappa$ are given in Theorem \ref{theo:lukExcursions}. The two missing cases for $\delta=0$ are not possible in the absorption model.}
 \label{tab:lukMeanders}
 \end{table}
\end{theo}
\vskip-1mm\begin{proof}[(Sketch)]
From \eqref{eq:genFuncEq2}, we get the bivariate generating function for meanders as
\begin{align}
 F(z,u) &= \frac{1 - z \left( P(u) - \PNgeq(u) \right) E(z)}{1-zP(u)}. \label{eq:lukBGForig} 
\end{align}
Hence, the generating function $M(z)$ for meanders is given by substituting $u=1$ in \eqref{eq:lukBGForig}:
\begin{align}
 \label{luk:FinAltFull}
 M(z) &:= F(z,1) = \frac{1}{1-z} - \left(1 - \PNgeq(1) \right) \frac{zE(z)}{1-z}. 
\end{align}
An elementary simplification of the last factor gives
\begin{align}
 \label{eq:lukMeandersParts}
 [z^n] M(z) &= 1 - \left(1 - \PNgeq(1) \right) \left([z^n] \frac{E(z)}{1-z} - [z^n]E(z) \right).
\end{align}
The asymptotics of the last term are known from Section \ref{sec:Luk}. For the function $\tfrac{E(z)}{1-z}$, we use \eqref{eq:lukExcursionsFinal} and elementary singularity analysis, 
like \cite[Fig. VI.5]{flaj09}. The result follows by distinguishing all different cases.
\end{proof}

\begin{remarkFormulaEnd} \label{rem:meandersReflAbs}
Formula \eqref{luk:FinAltFull} possesses a combinatorial interpretation:
a walk can only be absorbed after hitting the $x$-boundary, and at this place the walk is thus an excursion. 
Let~$e_n$ be the probability that a random walk of length $n$ is an excursion. A walk survives with probability $\PNgeq(1)$ and is killed with probability $(1-\PNgeq(1))$. The probability $m_{n+1}$ describing the number of meanders of length $n+1$ among all walks of length $n+1$ is given by all surviving walks of smaller length: 
$\displaystyle{ m_{n+1} = 1 - \left( 1 - \PNgeq(1) \right) \sum_{k = 0}^n e_k.}$ 
\end{remarkFormulaEnd}

\section{Final altitude of meanders}
\label{sec:finAlt}

The \emph{final altitude} of a path is defined as the ordinate of its endpoint. 

Let $X_n$ be the random variable associated to the final altitude of all meanders of length $n$. It satisfies
\begin{align}
 \label{eq:lukXnFinAlt}
 \PR[X_n = k] &= \frac{[z^n u^k] F(z,u)}{[z^n] F(z,1)}\,,
\end{align}
where $F(z,u)$ is the bivariate generating function for meanders from \eqref{eq:lukBGForig}.

\begin{theo}[Final altitude of meanders]
 \label{theo:lukMeanFinAlt}
 Consider the model of {\L}ukasiewicz walks. Let $\tau$ be the structural constant determined by $P'(\tau)=0$, $\tau>0$, $\delta = P'(1)$ be the drift and $\dNgeq = (\PNgeq)'(1)$ be the drift at $0$. 
The limit laws and the asymptotics of the expected final altitude of meanders for the reflection model are given in Table~\ref{tab:lukAsyMeanFinAltRef} and the ones for the absorption model are given in Table~\ref{tab:lukAsyMeanFinAltAbs}. 

 \begin{table}[htbp]
 \begin{center}
 \begin{tabular}{|c||Sc|Sc|Sc|}
 \hline & $\delta < 0$ & $\delta = 0$ & $\delta > 0$ \\
 \hline \hline Limit law & Discrete & Half-normal & Gaussian \\
 \hline\hline Supercritical & $\displaystyle{\E[X_n] \sim \frac{\dNgeq P''(1) + \delta {P_0^{\geq}}''(1)}{2 \delta (\delta - \dNgeq)}} $ & --- & \\
 \cline{1-3} Critical & --- & $\displaystyle{ \E[X_n] \sim\sqrt{\frac{2P''(1) n}{\pi}}}$& $\displaystyle{ \E[X_n] \sim\delta n}$ \\
 \cline{1-3} Subcritical & --- & --- & \\
 \hline
 \end{tabular}
 \end{center}
 \caption{Asymptotics of $\E[X_n]$ in the reflection model (i.e.,~$\PNgeq(1)=1$). The unfilled cases are not occurring under this model.} 
 \label{tab:lukAsyMeanFinAltRef}
 \end{table}

 \begin{table}[htbp]
 \begin{center}
 \begin{tabular}{|c||Sc|Sc|Sc|}
 \hline & $\delta < 0$ & $\delta = 0$ & $\delta > 0$ \\
 \hline \hline Limit law & Discrete & Rayleigh & Gaussian \\
 \hline\hline Supercritical & $\displaystyle{\E[X_n] \sim \left(1-\frac{1}{\rho_1}\right) \frac{E(1) F_u(\rho_1,1)}{E(\rho_1)}} $ & --- & \\
 \cline{1-3} Critical & $\displaystyle{\E[X_n] \sim \kappa \left(1-\frac{1}{\rho}\right) \frac{E(1) F_u(\rho,1)}{E(\rho)}}$ & --- & $\displaystyle{\E[X_n] \sim \delta n}$ \\
 \cline{1-3} Subcritical & $\displaystyle{\E[X_n] \sim r \left(1-\frac{1}{\rho}\right) \frac{ E(1) } { E(\rho) }}$ & $\displaystyle{\E[X_n] \sim \sqrt{\frac{ P''(1) \pi n}{2}}}$ & \\
 \hline
 \end{tabular}
 \end{center}
 \caption{Asymptotics of $\E[X_n]$ in the absorption model ($\PNgeq(1)<1$). The unfilled cases are not occurring under this model.
We denote by $F_u(\rho,1)$ the limit $z\to \rho$ and $u \to 1$ of the derivative of $F(z,u)$ with respect to $u$. Note that in $\frac{F_u(\rho,1)}{E(\rho)}$ the singularities at $z=\rho$ cancel and the limit exists. Furthermore, we have $r=F_u(\rho,1) - \frac{\delta \rho}{(1-\rho)^2}$.}
 \label{tab:lukAsyMeanFinAltAbs}
 \end{table}
\end{theo}

\begin{proof}[(Sketch)]
 Definition \eqref{eq:lukXnFinAlt} leads to the following formula for the expected value:
 \begin{align}
 \label{eq:lukExpFinalAlt}
 \E [X_n] = \frac{[z^n] \left. \frac{\partial}{\partial u} F(z,u) \right|_{u=1}}{[z^n] F(z,1)}.
 \end{align}
 Differentiating the Formula \eqref{eq:lukBGForig} for $F(z,u)$  with respect to $u$ yields
 \begin{align}
 \label{eq:lukAltMean}
 \left. \frac{\partial}{\partial u} F(z,u) \right|_{u=1} &= (\PNgeq)'(1) \frac{zE(z)}{1-z} + P'(1) \left( \PNgeq(1) - \PNgeq(u_1(z) \right) \frac{z^2 E(z)}{(1-z)^2}.
 \end{align}
 As a next step, we evaluate the $[z^n]$-operator term by term. 
 Firstly, note that the quotient $\tfrac{E(z)}{(1-z)^\beta}$ appears twice with $\beta=1$ and $\beta=2$. We use \cite[Theorem VI.12]{flaj09}, which gives the asymptotics for two generating functions with different radii of convergence, and previous results for $E$ like Theorem \ref{theo:lukExcursions}. 
 Secondly, the following lemma gives the behavior of the composition $\PNgeq(u_1(z))$:
 \begin{lemma}[A Puiseux behavior lemma]
 \label{lem:lukPNgequ1}
 Let $\PNgeq$ be the non-negative part of $P_0$, and $u_1$ be the small branch of the kernel equation in the {\L}ukasiewicz case. Then
 \begin{align*}
 \PNgeq(u_1(z)) &= 
 \begin{cases}
 \PNgeq(u_1(1)) - \alpha (1-z) + \frac{\alpha_2}{2} (1-z)^2 + \Landauo((1-z)^2), & \text{ for } \rho > 1,\\
 \PNgeq(1) - \kappa \sqrt{1-z} + \LandauO(1-z), & \text{ for } \rho = 1,
 \end{cases}
 \end{align*}
 with $\alpha = \left(\PNgeq \circ u_1\right)'(1) = -(\PNgeq)'(u_1(1))/P'(u_1(1))$ and $\alpha_2 = \left(\PNgeq \circ u_1\right)''(1)$.
 \end{lemma}
 
 The computations of the asymptotics of the expected value are then finalized by considering that 
the denominator of \eqref{eq:lukExpFinalAlt} is either 1 in the reflection model 
or its asymptotics is given by Theorem \ref{theo:lukMeanders} in the absorption model. 
 
 We now turn to the derivation of the underlying limit laws.
 Let $u$ be a fixed positive real number in $(0,1)$. Then the dominant singularity of $F(z,u)$ 
 is either $z=\rho$, the singularity of $E(z)$, or $z=1/P(u)$, the singularity of the denominator (compare \eqref{eq:lukBGForig}).
 Which one is the dominant one?   This depends on the value of the drift $\delta$:

 \begin{itemize}
 \item If the drift is negative the dominant singularity is found at $z=\rho$.
By the ``continuity theorem of discrete laws'', Theorem IX.$1$ from \cite{flaj09}, this leads to a discrete limit law.
 \item If the drift is zero the dominant singularity is at $z=\rho=1/P(1)=1$. 
In the reflection model, we show by the method of moment convergence the appearance of a half-normal distribution (see~\cite{Wallner} for more on this distribution).
In the absorption model, the application of the ``semi-large powers theorem'' (see~\cite{bfss01}),
 leads to a Rayleigh limit law, see also~\cite{DrSo97}. 
 \item If the drift is positive in the needed environment of $u=1$, the dominant singularity is found at $z=1/P(u)$ (this is due to $\rho>1$ in this case) and
 by the application of the ``quasi-powers theorem'' (Theorem IX.$8$ from \cite{flaj09}),
 this leads to a Gaussian limit law. 
 \end{itemize}
 The special role of the drift and an intuition of the underlying limit laws 
can be obtained by considering the influence of the drift on the expected number of meanders 
in Theorem \ref{theo:lukMeanders}.
\end{proof}

\section{Conclusion}
In this article, we investigated the generating functions of lattice paths on $\N$ 
with an absorbing or reflecting border at 0. 
To sum it up, our analysis of this reflection-absorption model can be divided into three parts. 
Firstly, we determined the general formula for the generating function by means of the kernel method. 
Secondly, we focused on {\L}ukasiewicz walks (i.e.,~the family of lattice paths in bijection with trees), 
and derived asymptotic results on the number of meanders and excursions, by means of singularity analysis and transfer theorems from analytic combinatorics. 
Thirdly, we investigated the limit laws for the returns to zero and the final altitude by utilizing schemes on generating functions which yield the convergence of the underlying distributions. 
Therein, we applied the ``continuity theorem of discrete laws''~\cite[Theorem IX.$1$]{flaj09}, the ``quasi-powers theorem''~\cite[Theorem IX.$8$]{flaj09}, and the ``semi-large powers theorem'' \cite{bfss01}.
Numerical values confirm our results (see figure below).
\begin{figure}[ht]
 \begin{center}
 {\includegraphics[width=0.6\textwidth]{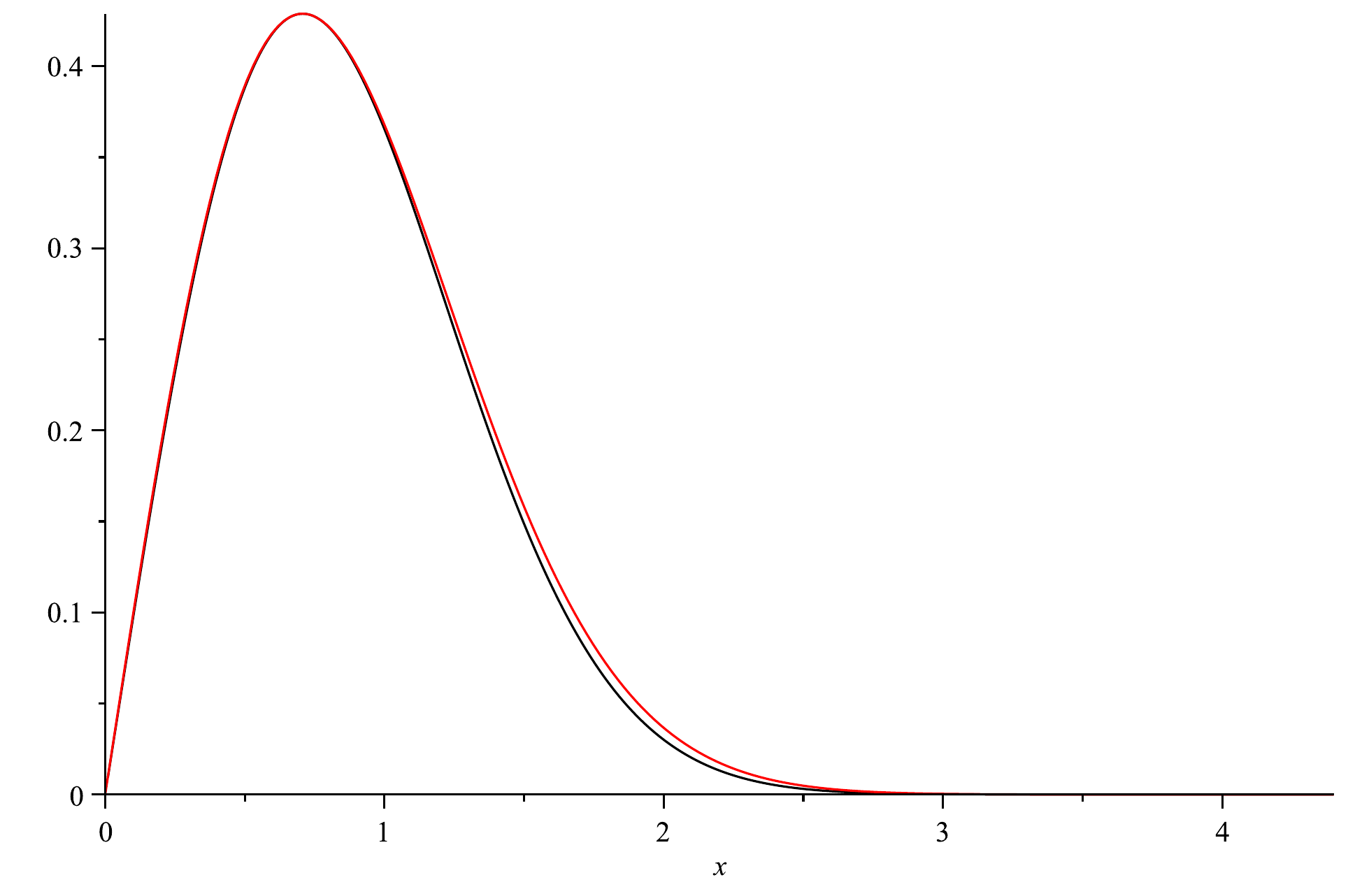}}

 \caption{In Theorem~\ref{theo:lawContacts}, we proved that the number of returns to 0 follows asymptotically a Rayleigh limit law (in the subcritical case of the absorption model). 
Our figure shows a perfect fit between the plot of the theoretical Rayleigh density (in black) and the plot of the empirical distribution (in red) for Motzkin paths of length $n=2000$.
The tiny discrepancy around $x=2$ is completely coherent with the error term, i.e.,~the speed of convergence, in $\LandauO(1/\sqrt{n})$.}
 \end{center}
\end{figure}

The situation is more tricky than what happens for the classical ``absorption model'' of lattice paths considered in~\cite{bafl02}:
one has to pay a price for introducing a more general model, as different cases (subcritical, critical and supercritical) 
have then to be distinguished and additional structural constants like $\dNgeq$ (the drift at $0$) play a key role. Putting it all together, there arose $9$ different cases for each model. 
Interestingly though, elementary considerations implied that some of these were impossible in the specific models (compare Table \ref{tab:lukAsyMeanFinAltRef} 
and Table \ref{tab:lukAsyMeanFinAltAbs} where non-existing cases are marked by a hyphen).
In the full version of this work, we give all details of the proofs omitted here
(we say a word on matters of periodicity, and we extend the results to excursions, meanders, walks, bridges, arches, 
beyond the {\L}ukasiewicz case, i.e.,~when the correspondence with trees do not hold anymore).
In another work in preparation, we give the asymptotics for some other harder parameters (area, height).

\phantomsection
\addcontentsline{toc}{chapter}{Acknowledgments}
\acknowledgements
\label{sec:ack}
This work is the result a collaboration founded by the SFB project F50 ``Algorithmic and Enumerative Combinatorics'' 
and the Franco-Austrian PHC ``Amadeus''. Michael Wallner is supported by the Austrian Science Fund (FWF) grant SFB F50-03 and by \"OAD, grant F04/2012. 
We also thank the two AofA referees and one additional referee for their feedback.
\phantomsection
\addcontentsline{toc}{chapter}{References}
\bibliographystyle{myabbrvnat}
\bibliography{aofa2014}
\label{sec:biblio}
\end{document}